%% file: main.tex
\documentclass{article}

\usepackage[english]{babel}

\usepackage[letterpaper,top=2cm,bottom=2cm,left=3cm,right=3cm,marginparwidth=1.75cm]{geometry}

\usepackage{amsmath}
\usepackage{amsthm}
\usepackage{amssymb}
\usepackage{graphicx}
\usepackage[colorlinks=true, allcolors=blue]{hyperref}
\usepackage{tikz}
\usepackage{tikz-cd}
\usetikzlibrary{calc, shapes, backgrounds, positioning}
\usepackage{tikzit}
\newtheorem{theorem}{Theorem}[section]
\newtheorem{corollary}{Corollary}[theorem]
\newtheorem{lemma}[theorem]{Lemma}
\newtheorem{proposition}[theorem]{Proposition}
\newtheorem*{conjecture}{Conjecture}

\theoremstyle{definition}
\newtheorem{definition}{Definition}[section]

\theoremstyle{remark}
\newtheorem*{remark}{Remark}
\newtheorem*{proofsketch}{Proof sketch}

\newlength\tindent
\title{Time complexity for deterministic string machines}
\author{Nur \c{C}ataltepe\thanks{MATS 5.0. Email: \texttt{ali@cataltepe.com}}, Vanessa Kosoy\thanks{The Hebrew University of Jerusalem.
Mentor for MATS 5.0. Email: \texttt{vanessa.kosoy@intelligence.org}}}
\date{May 9, 2024}
\begin{document}
\maketitle

\begin{abstract}
\noindent
Algorithms which learn environments represented by automata in the past have
had complexity scaling with the number of states in the automaton, which can
be exponentially large even for automata recognizing regular
expressions with
a small description length. We thus formalize a compositional 
language that can construct automata as transformations
between certain types
of category, representable as string diagrams, which better
reflects the description complexity of various automata. We
define complexity constraints on this framework by having them
operate on categories enriched over filtered sets, and using these
constraints, we prove
elementary results on the runtime and expressivity of a subset of these
transformations which operate deterministically on finite state spaces.
These string diagrams, or ``string machines,'' are themselves morphisms
in a category, so it is possible for string machines to create other
string machines in runtime to model computations which take more
than constant memory.
We prove sufficient conditions for polynomial
runtime guarantees on these, which can
help develop complexity constraints on string
machines which also encapsulate
runtime complexity.
\end{abstract}
\section{Introduction}
 String machines~\cite{kosoy2023compositional} is a framework that
 aims to provide a generalization of transducers, modeling them as
 functors that take morphisms in one category to ensembles 
 of ``recipes'' for constructing other morphisms in another category.
 We formalize and expand upon definitions for
 deterministic string machines sketched in the AI
 Alignment Forum post introducing the framework, with an eye towards
 introducing formalisms to help bound description complexity and
 runtime complexity. To this end, we introduce filtered
 transducers, which operate on categories where morphisms
 have ``degrees'' additive under composition, and prove some elementary
 runtime and expressivity properties for these transducers.

\subsubsection*{Acknowledgements}
This project was supported by a grant from the Berkeley Existential
Risk Initiative (BERI), as part of the ML Alignment \& Theory Scholars
(MATS) 5.0
research program.

\subsection{Motivations}
Theoretical reinforcement learning settings that utilize
automata (usually as transducers operating on 
histories)~\cite{ronca2022efficient} or Markov decision
processes to model
hypotheses about the environment often have time complexity
scaling in the number of states of the
environment~\cite{kearns1995computational}. However,
parameterized families of
automata with low description complexity can still grow exponentially
in size with these parameters. For example, given a parameterized
family of regular languages in some variable $n$, where the automaton
recognizing a given language has $O(1)$ states, the product automaton,
recognizing the intersection of the first $n$ languages, will have
$O(\exp(n))$ states, despite the fact that we could store this
parameterized family of automata with linearly increasing space
instead.
\\\\
Frameworks such as Krohn-Rhodes theory~\cite{krohn1965algebraic}
give decompositions of automata into algebraic structures in a way
that avoids exponential
growth of description
length, but these frameworks as they currently exist
do not generalize to automata more expressive
than deterministic finite state automata. Modeling
compositions of Mealy/Moore
machines with output alphabets that can themselves be interpreted
as Mealy/Moore machines could tractably model some intermediate
stages of expressivity between finite-state transducers and
Turing machines.
If an ML algorithm that operates
on an environment clearly more complex than can be
efficiently modeled by automata could be said to be ``approximating'' an
RL algorithm that operates on hypothesis spaces populated by automata,
then it must be taking advantage of some space-saving structure, as
ML algorithms have resource bounds too tight to scale with an
exponentially-growing number of states. It
would also be helpful to have a better approximation for
the Kolmogorov complexity of transducers for the purpose
of frugal priors in theoretical RL settings
similar to AIXI~\cite{hutter2000theory}, as the actual Kolmogorov
complexity is not computable in general.

\subsection{Prior work}
Category-theoretic generalizations of automata have been
explored before, either making
automata into algebras over certain
structures~\cite{adamekautomata} or morphisms in certain
categories~\cite{niu2023polynomial, boccali2023completeness}, or having
them operate on morphisms in a category~\cite{earnshaw2022regular,earnshaw2023string}.
We are interested in the latter two cases, which would give us
the ability to naturally model ``building other automata in runtime''
by modeling automata as morphisms between categories which themselves
includes automata as morphisms. Work by Earnshaw et al. on
monoidal languages~\cite{earnshaw2022regular,earnshaw2023string}
focuses on generalizing classes of structures
(string diagrams) beyond strings
which can be recognized by automata, but is not immediately
concerned with functions between classes of string diagrams.
Work by Niu~and~Spivak~\cite{niu2023polynomial} introduces
modes of composing natural transformations between representable
functors which induce families of maps that define
automata and transducers, and work on generalizations of
automata into morphisms in a category by 
Boccali~et~al.~\cite{boccali2023completeness} provides an
expressive framework for constructing a category of automata
in a way technically compatible with recursive automaton-building.
However, neither of these formalisms is especially focused on
finding more frugal or tractable ways of representing existing
computational processes.

\section{Deterministic transducers}
Here we provide necessary definitions for deterministic
string machines. Monadic string machines, which model
probabilistic behavior, are described in the AI Alignment Forum
post~\cite{kosoy2023compositional}, but the results in
Section~\ref{sec:paramcomplexity} are concerned solely with
the deterministic case.
\begin{definition}
\label{def:statecategory} Given a copy-discard category
$\mathcal{C}$, its corresponding
\textbf{state category} $\mathcal{C}^S$ is the category where
\begin{itemize}
    \item an object $A$ is a tuple $(S(A), V_A)$, where:
    \begin{itemize}
    \item \textbf{The state space of $A$} is a set $S(A)$.
    \item \textbf{The variables of $A$} are a functor $V_A$ from the
    discrete category corresponding to $S(A)$ to the discrete category
    formed by all finite ordered tuples of pairs of objects
    $X,Y$ in $\mathcal{C}$. We call the tuple associated with a given
    $x \in
    S(A)$ \textbf{the variables at $x$}, or $V_A(x)$. Since variables
    implicitly denote morphism signatures, we will refer to a given
    variable $(X,Y) \in V_A(x)$ as $X\to Y$. 
\end{itemize}
    \item a morphism $f:A \to B$ is a tuple
    $(S(f),V_f)$, where:
    \begin{itemize}
        \item \textbf{The state transition} $S(f)$ is a function
        $S(A) \to S(B)$.
        \item \textbf{The output function at $x \in S(A)$},
        denoted $V_f(x)$, is a collection of morphisms in a
        category freely-generated over
        $\mathcal{C}$, defined as follows:
        \\ \\
        For each variable $X_j'\to Y_j'$ at
        $S(f)(x)$,
        $V_f(x)$ defines a morphism $X_j'\to Y_j'$
        in the freely-generated
        category over $\mathcal{C}$ with generators
        provided by the variables
        $X_i\to Y_i$ in $V_A(x)$. The function
        between hom-sets in $\mathcal{C}$ can then be expressed by
        substituting each generator with a $\mathcal{C}$-morphism of
        the appropriate signature.
        \\ \\
        This induces a function between the product of
        the hom-sets $\mathrm{Hom}_{\mathcal{C}}(X_i,Y_i)$ for all
        variables $X_i\to Y_i$ at $x$ and the product of the hom-sets
        $\mathrm{Hom}_{\mathcal{C}}(X_j',Y_j')$ for all variables
        $X_j'\to Y_j'$ at $S(f)(x)$, which can be obtained by replacing
        each generator corresponding to a variable at $x$ with a 
        morphism of $\mathcal{C}$ of the appropriate signature.
    \end{itemize}
    \item Given two morphisms $f:A\to B$
    and $g:B \to C$, their composition
    $g\circ f:A\to C$ is a tuple $(S(g\circ f),V_{g \circ f})$ where:
    \begin{itemize}
        \item The composite
        state transition
        $S(g\circ f)$ is defined as $S(g) \circ S(f)$.
        \item The composite output function $V_{g\circ f}(x)$ is
        defined as follows:
        For each variable $X_k''\to Y_k''$ defined at $S(g\circ f)(x)$,
        $V_g(S(f)(x))$ defines a morphism
        in the category freely generated
        over $\mathcal{C}$ with generators corresponding to the variables
        $\{X_j'\to Y_j'\}_j$
        at $S(f)(x)$. Now, $V_f(x)$ induces a functor
        from $\mathcal{C}[X_1'\to Y_1',\dots]$ to
        $\mathcal{C}[X_1\to Y_1,\dots]$ (where $\{X_i\to Y_i\}_i$ denote
        the variables at $x$), which fixes
        $\mathcal{C}$ but replaces every free generator of the category
        freely generated over the variables at $S(f)(x)$ with the
        morphism associated with it
        by $V_f(x)$,
        in the category freely generated over the variables
        at $x$. Therefore, the output function $V_{g\circ f}(x)$ at
        $x$ is simply the image of each of the morphisms generated
        by $V_g(x)$ under this functor defined by $V_f(x)$. The function
        between products of
        hom-sets induced by the composite output function is
        therefore just the composition of the functions induced by the
        individual morphisms' output functions.
    \end{itemize}
    \item Given two objects $A,B$, their monoidal
    product $A\otimes B$ has state space $S(A\otimes B)=
    S(A)\times S(B)$ and variables $V_{A\otimes B} = 
    V_A \times V_B$, in the sense that the variables at
    $(x,y) \in S(A\otimes B)$ are the concatenation of
    the variables lists $V_A(x)$ and $V_B(y)$. The monoidal identity
    is the single-element state space with no variables. There is a
    copy-discard (in fact, a Markov category, since a single-element
    state space with no variables can only have a single output
    function into it and is therefore a terminal object) structure
    on $\mathcal{C}^S$ provided by a morphism which is the copy map
    on the state space and for each variable $X_i\to Y_i$ in
    the output state space, generates a morphism solely consisting
    of the free generator of that same signature that that variable
    is a clone of.
\end{itemize}

The state category construction could be thought of as a
functor on the category of small categories.
Given a functor $F:\mathcal{C}\to\mathcal{D}$,
we can define a functor $F^S:\mathcal{C}^S\to\mathcal{D}^S$
with maps objects $(S(A),V_A)$ of $\mathcal{C}^S$ to
objects $(S(F^S(A)),V_{F^S(A)})$ of $\mathcal{D}^S$ such that
$S(F^S(A)) \simeq S(A)$ as sets with an
isomorphism $\varphi_{F^S,A}$, and for every $x \in S(A)$,
$V_{F^S(A)}(\varphi_{F^S,A}(x))= F(V_A(x))$, in the sense that the
signatures in
$V_{F^S(A)}(\varphi_{F^S,A}(x))$
are simply the image under $F$ of the signatures
at $V_A(x)$. The image of a state-category morphism $A\to B$
$(S(f),V_f)$ in $\mathcal{C}^S$ under $F^S$ is a morphism
with an identical state transition such that the relations between
the variable tuples provided by the isomorphism between sets
commutes with this state transition, that is,
$\varphi_{F^S,B} \circ S(f) = S(F^S(f)) \circ \varphi_{F^S,A}$.
For every $x\in S(F^S(A))$, the
morphisms generated by the output function $V_{F^S(f)}(x)$
are simply the images, under the natural extension of $F$ to the
appropriate freely-generated categories over $\mathcal{C}$
and $\mathcal{D}$, of the morphisms generated by the output function
$V_{f}(\varphi_{F^S,A}^{-1}(x))$.
\\ \\
However, it is unclear if
the state category construction actually
provides a 2-functor. Given functors $F,G:\mathcal{C}
\to\mathcal{D}$ and a natural transformation $\eta:F\to G$, we
try defining a natural transformation $\eta^S:F^S\to G^S$ where,
for an object $A$ of $\mathcal{C}^S$, the morphism
$\eta^S_A$ has a bijection $S(F^S(A)) \to
S(G^S(A))$ (we know this bijection exists, since
$S(F^S(A))$ and $S(G^S(A))$ are bijective to the
same set), which commutes with
$\varphi_{F^S,A}$ and $\varphi_{G^S,A}$, as its state transition.
An issue arises when
trying to construct an output function for
$\eta_A^S$. For a given $x \in S(A)$, where
$V_A(x)$ consists of signatures $X_i\to Y_i$ and
therefore $V_{F^S(A)}(\varphi_{F^S,A}(x))$ consists of signatures
$F(X_i) \to F(Y_i)$ (and likewise for $G^S$), the morphisms provided
by $\eta_{X_i}:F(X_i)\to G(X_i)$ and $\eta_{Y_i}:
F(Y_i)\to G(Y_i)$ and the generators of signature
$F(X_i) \to F(Y_i)$ do not guarantee a way to construct a morphism
of signature $G(X_i)\to G(Y_i)$, since we are not provided with
any natural choices for outgoing morphisms from $G(X_i)$.
\\ \\
With the state category defined, we can now define the primary
building block of a string machine, the transducer.
\end{definition}
\begin{definition}
    \label{def:dettransducer}
    A \textbf{deterministic transducer} is a collection of the
    of the following data:
    \begin{itemize}
        \item an \textbf{input category} $\mathcal{C}$ and
        \textbf{output category} $\mathcal{D}$, which are both
        copy-discard categories,
        \item a \textbf{primary input ($\mathcal{C}$-)signature}
        $A\to B$,
        \item \textbf{auxiliary input ($\mathcal{D}$-)signatures}
        $X_1\to Y_1,\dots,X_n\to Y_n$,
        \item \textbf{output ($\mathcal{D}$-)signatures}
        $X_1'\to Y_1',\dots, X_m'\to Y_m'$,
        \item a \textbf{structure functor} $F:\mathcal{C}
        \to\mathcal{D}^S$ satisfying the following properties:
        \begin{itemize}
            \item $F$ is strong monoidal.
            \item For every $x\in S(F(A))$, the
            variables at $V_{F(A)}(x)$ are a prefix of
            $X_1\to Y_1,\dots, X_n \to Y_n$.
            \item For every $y \in S(F(B))$, 
            $X_1'\to Y_1',\dots, X_m'\to Y_m'$ is a prefix
            of $V_{F(B)}(y)$.
        \end{itemize}
    \end{itemize}
    The structure functor
    additionally induces an \textbf{input state space}
    $S(F(A))$ and \textbf{output state space} $S(F(B))$.
\end{definition}
Thus, a deterministic transducer takes a morphism
in $\mathcal{C}$ of signature
$A\to B$ to a ``finite state machine-controlled''
way of constructing morphisms in $\mathcal{D}$, and outputs
morphisms in $\mathcal{D}$ by substituting the free
generators over $\mathcal{D}$ with the morphisms received in
the auxiliary inputs.

\subsection{Composition of transducers}
If we consider a transducer to have incoming legs labeled with
its input state space and primary-auxiliary input signatures, and
outgoing legs labeled with its output state space and output signatures,
we can contract incoming legs with outgoing legs of the same signature
such that no cycles are formed if we interpret the transducer/leg complex
as a directed graph.
Now, given inputs to the remaining free input legs, the contracted
input legs will receive the outputs generated at their contracted
output legs. We can also take the ``product'' of two transducers
(or compositions of transducers) and have them each process their
inputs and produce their outputs independently.
\begin{definition}
    \label{def:strmachinecategory} The
    category $\mathbf{StrMach}$
    of string machines is defined as follows:
    \begin{itemize}
        \item Objects are finite tuples of the form
        $((S(A_1),\dots,S(A_n)),(X_1\to Y_1,\dots,
        X_m\to Y_m)$, where each $S(A_i)$ is a set and
        each $X_j\to Y_j$ is a morphism signature in some
        copy-discard category. The morphism signatures are not
        necessarily in the same category.
        \item Morphisms are \textbf{string machines}, which
        consist of transducers and finite compositions of transducers.
        The two objects that a string machine is a morphism between
        is determined by its total input and output signatures.
        \item There exists a monoidal product on $\mathbf{StrMach}$
        defined as the concatenation of the tuples making up objects
        (with the identity given by an empty tuple),
        and a natural copy-discard structure provided by a
        string machine with two state spaces (the
        output being the self-product of the input)
        and a pre-loaded state transition that is the copy morphism
        in the state category.
    \end{itemize}
\end{definition}

\subsection{The meta-vertex and looping}
Since $\textbf{StrMach}$ is itself a copy-discard category, we can
now have string machines which output other string machines. To lessen
reliance on closed categories, we will define a hierarchy of categories
corresponding to ``levels of meta.''
\begin{definition}
    \label{def:metavertex}
    Given objects $X,Y$ in $\mathbf{StrMach}$, we define
    the \textbf{meta-vertex} to be a special
    kind of string machine with input signature
    corresponding to the data in $X$'s definition as
    an object in $\mathbf{StrMach}$ alongside $X\to Y$, and output
    signature corresponding to the data in $Y$'s definition. The
    meta-vertex therefore takes a string machine with input
    signature $X$ and output signature $Y$ as well as the necessary
    input data for the string machine as input, and outputs that
    string machine's output as its output. Fixing some partially-ordered
    set $(S,\leq)$, we define the category $\mathbf{StrMach}_0=
    \mathbf{StrMach}$ (which has no meta-vertices) for any
    element $0 \in S$ such that $k\leq 0 \implies k=0$ for
    any $k\in S$,
    and recursively define
    the category $\mathbf{StrMach}_k$ for some
    $k \in S$ to be the category generated over
    $\mathbf{StrMach}$ by adding only meta vertices which 
    accept string machines that are morphisms in
    $\mathbf{StrMach}_{k'}$ for $k' \leq k$.
\end{definition}

We can thus have a string
machine consisting of a transducer that builds transducers that
themselves build transducers (and so on), and feeds this transducer
into a meta-vertex, to create a loop which is guaranteed to terminate if
we fix our ordinals to come from $\mathbb{N}$.

\section{Parameterizing complexity in
deterministic string machines}
\label{sec:paramcomplexity}
\subsection{Filtered transducers}
\begin{definition}
    \label{def:filtsetcat} The
    category $\mathbf{FiltSet}$ of filtered sets is the category such that
    \begin{itemize}
    \item an object is a tuple $(S,\deg_S)$, where $S$ is a set and
    $\deg_S:S \to \mathbb{N}$ is a function,
    \item a morphism $f:(S,\deg_S) \to (T,\deg_T)$ is a function
    $S \to T$ such that $\deg_T(f(s)) \leq \deg_S(s)$ for all $s\in S$. 
    \end{itemize}
\end{definition}
We will generally refer to objects in $\mathbf{FiltSet}$ solely
by the symbol corresponding to the underlying set going forward. One
can observe that the identity function on a set $S$ by definition satisfies
$\deg_S(\mathrm{id}_S(s))=\deg_S(s)$ for all $s \in S$ and is thus a
morphism in $\mathbf{FiltSet}$. One can also observe that given
$f:S\to T$ and $g:T \to V$,
$\deg_V(g(f(s))) \leq \deg_T(f(s)) \leq \deg_S(s)$ for all
$s \in S$, and therefore $g \circ f$ is also a morphism in
$\mathbf{FiltSet}$.
Therefore, $\mathbf{FiltSet}$ is indeed a category.

\begin{definition}
\label{def:filtsetproduct}
Given two objects $S,T \in \mathrm{Ob}(\mathbf{FiltSet})$,
we define
their \textbf{filtered product}
$S \otimes T$ to be the set $S \times T$ equipped
with the function $\deg_{S \otimes T}:S \times T \to \mathbb{N}$ satisfying
$\deg_{S\otimes T}(s,t) = \deg_S(s)+\deg_T(t)$ for all $(s,t) \in 
S\times T$. Given a morphism $f:S \to U$ and a morphism
$g:T\to V$, we define the morphism
$f\otimes g:S\otimes T \to U\otimes V$ to be the function $f\times g$.
Indeed, we have that $\deg_{U\otimes V}(f(s),g(t)) = 
\deg_U(f(s))+\deg_V(g(t))\leq \deg_S(s)+\deg_T(t) = \deg_{S\otimes T}(s,t)$,
so $f\otimes g$ is a morphism in
$\mathbf{FiltSet}$.
\end{definition}
Due to the associativity and commutativity of addition, as well
as the natural associativity and commutativity (up to isomorphisms which
are still isomorphisms in $\mathbf{FiltSet}$) of the cartesian product,
$-\otimes -$ is naturally associative and commutative up to isomorphism.
Additionally, the one-element set $\mathbf{1}$ equipped with
$\deg_{\mathbf{1}}(\cdot)=0$ and unitor maps which are the same as in
$\mathbf{Set}$ (which are, by their definition, filtered morphisms)
provides a
left and right unit for $-\otimes -$, making
$\mathbf{FiltSet}$ a symmetric monoidal category.

\begin{remark}
Suppose filtered sets $S,T,U$ and filtered morphisms $f:S \to T$
and $g:S\to U$. Then, the unique factoring function $S\to T\times U$
defined by $s \mapsto (f(s),g(s))$ is only a filtered morphism 
$S \to T\otimes U$ if
$\deg_T(f(s))+\deg_U(g(s)) \leq \deg_S(s)$, which does not hold in
general. Therefore, $-\otimes -$ does not provide a product except for
when at least one of the sets has degree uniformly zero. However,
$\mathbf{FiltSet}$ does have finite products $S\times T$ where
$\deg_{S\times T}(s,t) := \max(\deg_S(s),\deg_T(t))$. We will not be using
this construction.
\end{remark}

\begin{remark}
The set-theoretic disjoint union,
with its degree function
being the canonical factoring map to $\mathbb{N}$ of its components'
degree functions, provides all finite coproducts in
$\mathbf{FiltSet}$.
\end{remark}

\begin{definition}
    \label{def:filteredmorphism} A \textbf{filtered-morphism
    category}
    $\mathcal{C}$ is a locally small symmetric
    monoidal category enriched over
    $\mathbf{FiltSet}$, using $\mathbf{FiltSet}$'s filtered product
    $-\otimes -$ as its monoidal structure.
\end{definition}
This expresses the notion of morphisms having degrees which
are subadditive under composition in a way that naturally extends to a
complexity constraint on transducers. As the monoidal identity of
$\mathbf{FiltSet}$ is the single-element set with degree zero, the
arrows $I_{\mathbf{FiltSet}} \to \mathrm{Hom}_{\mathcal{C}}(A,A)$ providing
the identity 
morphism $\mathrm{id}_A$
in the enrichment construction will ensure that identity morphisms are
degree zero.
\\ \\
One can generalize the construction of $\mathbf{Filtset}$ by replacing
$\mathbb{N}$ with any partially-ordered symmetric monoid
(where the partial order is, obviously, translation-invariant
as it is with the usual definition of a partially-ordered
group). We will make use of
categories filtered over $\mathbb{N}^2$ (with the partial order
given by $(a,b) \leq (c,d)$ if and
only if $a\leq c$ and $b\leq d$) in our
construction of the filtered state category, as follows:

\begin{definition}
    \label{def:freegenfilteredmorphism} Given a filtered-morphism category
    $\mathcal{C}$ and a list $\{(X_i \to Y_i, a_ix+b_i)\}_{i}$ of
    $\mathcal{C}$-signatures (generators)
    and linear polynomials with coefficients
    in $\mathbb{N}$, the \textbf{freely-generated $\mathbb{N}^2$-filtered
    category over $\mathcal{C}$} with this list of signatures,
    or $\mathcal{C}[X_1\to Y_1,\dots, X_n \to Y_n]$, is the
    symmetric monoidal category
    where:
    \begin{itemize}
        \item objects are exactly the objects of $\mathcal{C}$.
        \item A morphism $f:X\to Y$ is either a morphism of
        $\mathcal{C}$, or obtained by composing morphisms of
        $\mathcal{C}$ with free generators added to
        $\mathrm{Hom}_{\mathcal{C}}(X_i,Y_i)$ for each
        $\mathcal{C}$-signature $X_i\to Y_i$ in the list.
        \item The \textbf{$\mathbb{N}^2$-degree} of a morphism $f$ is
        a linear polynomial with coefficients in
        $\mathbb{N}$, which is either its
        degree in $\mathcal{C}$ (expressed as a constant)
        if it is a $\mathcal{C}$-morphism, the polynomial $a_ix+b_i$ if
        it is one of the generators $X_i \to Y_i$ added when creating
        the category, or the sum of the degrees of all morphisms involved
        in its composition, with the degree
        contributions of composition-chains of
        $\mathcal{C}$-morphisms being the degrees in $\mathcal{C}$ of the
        morphism they compose to, and these being strictly additive 
        with each other and the degrees of any generator morphisms once
        all simplifications have been performed.
        \\ \\
        The monoidal product of
        a morphism with generators in its composition-chain and a
        $\mathcal{C}$-morphism is taken to have the sum of their
        $\mathbb{N}^2$-degrees.
    \end{itemize}
\end{definition}
We can now introduce filtered state categories.
\begin{definition}
\label{def:filteredstatecategory} Given a filtered-morphism category
$\mathcal{C}$, its corresponding
\textbf{filtered state category} $\mathcal{C}^S$ is the category where
\begin{itemize}
    \item an object $A$ is a tuple $(S(A), V_A, \deg_A)$, where:
    \begin{itemize}
    \item \textbf{The state space of $A$} is a set $S(A)$.
    \item \textbf{The variables of $A$} are a functor $V_A$ from the
    discrete category corresponding to $S(A)$ to the discrete category
    formed by all finite ordered tuples of pairs of objects
    $X,Y$ in $\mathcal{C}$. We call the tuple associated with a given $x \in
    S(A)$ \textbf{the variables at $x$}, or $V_A(x)$. Since variables
    implicitly denote morphism signatures, we will refer to a given
    variable $(X,Y) \in V_A(x)$ as $X\to Y$. 
    \item \textbf{The degree of a variable $X\to Y \in V_A(x)$} for some
    $x \in S(A)$, denoted $\deg_{A,x}(X \to Y)$ (or simply
    $\deg_{A}(X\to Y)$ if $x$ can be inferred from context) is a 
    tuple $(a,b) \in \mathbb{N}^2$. We impose the same partial ordering
    $\leq$
    on these tuples as we do with degrees in the freely-generated
    $\mathbb{N}^2$-filtered categories.
\end{itemize}
    \item a morphism $f:A \to B$ of degree $\ell$, where
    $\ell \in \mathbb{N}$, is a tuple
    $(S(f),V_f)$, where:
    \begin{itemize}
        \item \textbf{The state transition} $S(f)$ is a function
        $S(A) \to S(B)$.
        \item \textbf{The output function at $x \in S(A)$},
        denoted $V_f(x)$, is a collection of morphisms in a
        freely-generated $\mathbb{N}^2$-filtered category over
        $\mathcal{C}$, defined as follows:
        \\ \\
        For each variable $X_j'\to Y_j'$ of degree $a_j'x+b_j'$, at
        $S(f)(x)$,
        $V_f(x)$ defines a morphism $X_j'\to Y_j'$ of degree
        at most $a_j'(x+\ell)+b_j'=a_j'x+(a_j'\ell+b_j')$ (that is,
        the degree of the morphism cannot exceed this polynomial in either
        of its terms)
        in the freely-generated
        $\mathbb{N}^2$-filtered category over $\mathcal{C}$ with generators
        and the degrees of those generators provided by the variables
        $X_i\to Y_i$ of degrees $a_ix+b_i$ in $V_A(x)$. The function
        between hom-sets in $\mathcal{C}$ can then be expressed by
        substituting each generator with a $\mathcal{C}$-morphism of
        the appropriate signature.
        \\ \\
        This induces a function (but not a filtered
        one!) between the product of
        the hom-sets $\mathrm{Hom}_{\mathcal{C}}(X_i,Y_i)$ for all
        variables $X_i\to Y_i$ at $x$ and the product of the hom-sets
        $\mathrm{Hom}_{\mathcal{C}}(X_j',Y_j')$ for all variables
        $X_j'\to Y_j'$ at $S(f)(x)$, which can be obtained by replacing
        each generator corresponding to a variable at $x$ with a 
        morphism of $\mathcal{C}$ of the appropriate signature.
    \end{itemize}
    \item Given two morphisms $f:A\to B$ of degree $\ell$
    and $g:B \to C$ of degree $\ell'$, their composition
    $g\circ f:A\to C$ is a tuple $(S(g\circ f),V_{g \circ f})$ where:
    \begin{itemize}
        \item The composite
        state transition
        $S(g\circ f)$ is defined as $S(g) \circ S(f)$.
        \item The composite output function $V_{g\circ f}(x)$ is
        defined as follows:
        For each variable $X_k''\to Y_k''$ defined at $S(g\circ f)(x)$, of
        degree $a_k''x+b_k''$,
        $V_g(S(f)(x))$ defines a morphism
        of degree at most $a_k''(x+\ell')+b_k''=a_k''x+(a_k''\ell'+b_k')$
        in the category freely generated
        over $\mathcal{C}$ with generators corresponding to the variables
        $\{X_j'\to Y_j'\}_j$
        at $S(f)(x)$ and their degrees. Now, $V_f(x)$ induces a functor
        from $\mathcal{C}[X_1'\to Y_1',\dots]$ to
        $\mathcal{C}[X_1\to Y_1,\dots]$ (where $\{X_i\to Y_i\}_i$ denote
        the variables at $x$), which fixes
        $\mathcal{C}$ but replaces every free generator of the category
        freely generated over the variables at $S(f)(x)$ with the
        morphism associated with it
        by $V_f(x)$,
        in the category freely generated over the variables
        at $x$. Therefore, the output function $V_{g\circ f}(x)$ at
        $x$ is simply the image of each of the morphisms generated
        by $V_g(x)$ under this functor defined by $V_f(x)$. The function
        between products of
        hom-sets induced by the composite output function is
        therefore just the composition of the functions induced by the
        individual morphisms' output functions.
        \\ \\
        Now, for a given
        variable $X_k''\to Y_k''$ at $S(f\circ g)(x)$ of degree
        $a_k''x+b_k''$,
        every generator of some degree
        $a_j'x+b_j'$
        involved in the morphism of
        (initial) degree $a_k''x+a_k''\ell'+b_k''$
        defined
        by $V_g(x)$ has been replaced with one of degree
        $a_j'x+a_j'\ell+b_j'$. Suppose each generator $X_j'\to Y_j'$
        appears with multiplicity $m_j'$ in the morphism corresponding to
        $X_k''\to Y_k''$. Then, we have it that $\sum_{j}m_ja_j' \leq 
        a_k''$, and since $a_j'\ell$ worth of degree has been added to
        the morphism for every generator appearing in it, we have that
        the degree of the morphism generated by $V_{g\circ f}(x)$ must
        be $a_k''+a_k''\ell'+b_k''+\sum_{j}m_ja_j'\ell \leq
        a_k''+a_k''(\ell+\ell')+b_k''$. So, the
        degree of $g\circ f$ is $\ell+\ell'$.
    \end{itemize}
\end{itemize}

\end{definition}

\begin{remark}
$\mathcal{C}^S$ is a filtered-morphism category with the degrees
of morphisms being defined as in the definition of
the category,
as the identity morphism of each object $A$ has an output function that
outputs a variable of degree $ax+b$ (the generator
corresponding to the variable itself) for any variable of degree
$ax+b$ and is thus of degree zero. Additionally, as also demonstrated
in the definition, degree is subadditive with composition, meaning
that the composition arrow is indeed a morphism in
$\mathbf{FiltSet}$.
\end{remark}

\begin{definition}
\label{def:filteredfunctor}
A \textbf{filtered functor} $F:\mathcal{C}\to \mathcal{D}$ is an
enriched functor
between filtered-morphism categories.
\end{definition}

Given a filtered functor $F:\mathcal{C} \to \mathcal{D}$, since
the definition of an enriched functor requires the functor's action
on morphisms 
to be provided by a morphism between hom-objects within the category,
it automatically holds that $\deg_{\mathcal{D}} F(\alpha)
\leq \deg_{\mathcal{C}}(\alpha)$ for all morphisms
$\alpha$ in
$\mathcal{C}$.

\begin{definition}
\label{def:filtereddcc} A \textbf{filtered deterministic transducer}
is a deterministic transducer where the input and output categories,
$\mathcal{C}$ and $\mathcal{D}$, are
filtered-morphism categories, and the functor $F:\mathcal{C}\to\mathcal{D}^S$
affording the transducer
structure is a
filtered functor to the filtered state category of the output category.
\end{definition}
\subsection{Finite-state string machines}
For the following section, we assume
a filtered deterministic transducer $T$ with input category
$\mathcal{C}$, output category $\mathcal{D}$, structure functor
$F$, input state space $S(A)$, output state space $S(B)$, primary
input $\mathcal{C}$-signature $X\to Y$, auxiliary input
$\mathcal{D}$-signatures $(X_i \to Y_i)_{1 \leq i \leq n}$, and
output $\mathcal{D}$-signatures $(X_j' \to Y_j')_{1 \leq j \leq m}$.
We will now demonstrate that requiring a finite output state space on
a transducer provides us with certain output size and runtime guarantees.

\begin{proposition}
\label{prop:boundedduplication}
If $S(B)$ is finite, and, for
every $x \in S(A)$,
the linear term of the degree of every variable in $V_{A}(x)$ is
nonzero, then, for every $x \in S(A)$ and every primary input
$\mathcal{C}$-morphism $\alpha$, the number of copies of
the generator corresponding to each
auxiliary input $X_i \to Y_i$ appearing in each morphism
corresponding to each output signature
$X_j'\to Y_j'$ in
$V_{F(\alpha)}(S(F(\alpha))(x))$ is bounded by a
fixed integer $k$, which does not depend on $\deg (\alpha)$.
\end{proposition}
\begin{proof}
Since $S(B)$ is finite, and the number of variables which exist
at any state is finite,
there exists a variable with maximal linear degree $a'_{\max}$ across
all variables in all states in $S(B)$, in
some state in $S(B)$. No morphism in $\mathcal{D}^S$ can generate
a morphism in the appropriate $\mathbb{N}^2$-graded free category
over $\mathcal{D}$ with linear degree greater than that of its corresponding
variable. Therefore, given any starting state $x \in S(A)$, no morphism
$f:A\to B$
in $\mathcal{D}^S$ exists where any component
of the output function has more than
$a'_{\max}$ copies of any of the generators corresponding to the variables
at
$x$, since each of these will have sizes with linear term at least $1$.
\end{proof}

\begin{proposition}
    \label{prop:finstateslinearoutput}
    Suppose $S(B)$ is finite,
    and every variable in every state in $S(A)$ has
    a linear term of $1$ or greater in its degree.
    Denote the primary input
    morphism of $T$ as $\alpha$
    and its tuple of auxiliary input morphisms as
    $\beta_1,\dots,\beta_n$, and define
    $\deg_m(\beta) := \max_{1\leq i \leq n}(\deg(\beta_i))$. Then,
    for every output signature $X'_j\to Y_j'$, there exist
    fixed integers
    $a^{T}_j, b^T_k, c^T_j$ such that, for the morphism
    $\gamma_j$ outputted for this signature given these inputs,
    $\deg(\gamma_j) \leq a^T_j \deg(\alpha)+b^T_j \deg_m(\beta)+c^T_j$.
\end{proposition}
\begin{proof}
For every
output signature $X_j'\to Y_j'$, there exists a state $y \in S(B)$
such that $\deg_{B,y}(X_j'\to Y_j')$
attains some maximal linear
coefficient $a_{\max}'$, the generators corresponding
to auxiliary inputs of $T$ can appear in $X_j'\to Y_j'$
at most some fixed
$a_{\max}'$ number of times. Additionally,
a similar maximum $b_{\max}'$ can be found for the constant term
of $\deg_{B,y}(X_j'\to Y_j')$. Since, by the definition of a filtered
functor, $\deg(F(\alpha)) \leq \deg(\alpha)$, and by the definition
of the state category, the $\mathbb{N}^2$-degree of the output at
$X_j'\to Y_j'$ can be at most $a_j'x+a_j'\deg(F(\alpha))+b_j'$, then
$\deg(\gamma_j) \leq a_{\max}'\deg_m(\beta)+a_{\max}'\deg(\alpha)+
b_{\max}'$.
\end{proof}

\begin{corollary}
\label{corr:linearruntime} Suppose a string machine composed
only of a finite set of filtered transducers $\{T_k\}_{1\leq k
\leq N}$, where
each transducer $T_k$ has input state space $S(A_k)$ and output
state space $S(B_k)$ such that every variable at every state in
$S(A_k)$ has a degree with $1$ or greater linear term, and
$S(B_k)$ is finite. Denote by $\alpha_k'$ the morphism that
$T_k$ will receive as its primary input morphism, given overall
input morphisms $\alpha_1,\dots,\alpha_n$ and
input states $x_1,\dots,x_m$ to the string machine.  Then,
$\sum_{k=1}^N \deg(\alpha_k') = O\left(\sum_{i=1}^n
\deg(\alpha_i)\right)$.
\end{corollary}

\begin{proof}
Suppose two transducers $T_1,T_2$, where $T_1$ receives a primary
input morphism $\alpha$ and auxiliary input morphisms
$\beta_1,\dots,\beta_m$. Therefore, by Proposition~\ref{prop:finstateslinearoutput}, there exist integers
$a_1,b_1,c_1$ such that, for a given output $\gamma$ produced by $T_1$,
the degree of $\gamma$ satisfies
$\deg(\gamma) \leq a_1\deg(\alpha)+b_1\deg_m(\beta)+c_1$.
Now, if the primary input of $T_2$ is not connected to an output of
$T_1$, then the string machine composed of these transducers satisfies
the proposition, as the inputs of $T_2$ are now included in the overall
inputs of the string machine. If $\gamma$ is composed into the primary
input of $T_2$, however, then $\deg(\alpha_2') \leq
a_1\deg(\alpha)+b_1\deg_m(\beta)+c_1$. Given similar coefficients
$a_2,b_2,c_2$ for a given output $\gamma_2$ of $T_2$, we can see that
$\deg(\gamma_2) \leq a_2(\deg(\alpha_2'))+b_2\deg_m(\beta_2)+c_2$, where
$\deg_m(\beta_2)$ is bounded from above by either the same kind of
linear sum as the one for $\gamma$ (if the greatest degree among the
auxiliary inputs comes from an output of $T_1$), or the degree of
yet another additional input to the overall string machine. Therefore,
any further composition of transducers $T_3,\dots$ will result in
$\alpha_3',\dots$ also being bounded by a linear sum of the
degrees of the overall
inputs of the string machine, meaning that the sum
$\sum_{k=1}^N \deg(\alpha_k')$ is also bounded by a linear sum of these
degrees.
\end{proof}
Corollary~\ref{corr:linearruntime} can be used to bound the runtime of
string machines, if, fixing some model of computation, we
have a bound on the time it takes to run a single transducer as
a function of its inputs. We note
that the coefficients bounding output size are multiplied together upon
composition into either an auxiliary input or a primary input, but
only composition into a primary input ties the degree
of the state transition that will occur in the
second transducer to the primary input degree of the first
transducer. We define the sum bounded in the corollary as follows:
\begin{definition}
    \label{def:totalinput} Take a string machine $K$ with input
    signatures $X_1\to Y_1, \dots X_n\to Y_n$ and input state
    spaces $S(A_1),\dots,S(A_m)$, where the transducers forming
    $K$ form the set $\{T_k\}_{1\leq k \leq N}$. Denote
    by $\alpha_k'$ the morphism that the transducer $T_k$
    will receive as its primary input morphism, given
    inputs $\alpha_1,\dots,\alpha_n$
    ($\overline{\alpha}$ for short) and $x_1,\dots,x_m$ 
    ($\overline{x}$ for short) to $K$.
    We define the
    \textbf{internal total primary input degree} of $K$ to be
    $\mathrm{IPD}_K(\overline{\alpha},\overline{x}):=\mathrm\sum_{k=1}^N \deg(\alpha_k')$. 
\end{definition}
While we have established that, lacking a meta-vertex,
$\mathrm{IPD}_K$ will be linearly bounded
by the total degree of its
inputs, its scaling properties as a function of
the number of transducers composed has no such guarantees. One can
observe that composing more and more transducers that double
their input size together causes exponential growth
in $\mathrm{IPD}_K$. As we are
interested in the runtime properties of both string machines that
incorporate a meta-vertex and of parameterized families of string
machines, we must obtain some sufficient conditions for
$\mathrm{IPD}_K$ to
grow polynomially in the number of transducers in
our string machine. To help determine these conditions,
we will define the output-degree-bounding parameters referenced in
Proposition~\ref{prop:finstateslinearoutput}:
\begin{definition}
    \label{def:outputdegree} Take a filtered deterministic transducer
    $T$
    with a finite output state space $S(B)$ and an output signature
    $X \to Y$ of $T$. Given an input state $x$, primary input
    morphism
    $\alpha$, and auxiliary input morphisms $\beta_1,\dots,\beta_m$
    to $T$,
    denote by $\gamma$ the output morphism produced for $X\to Y$.
    We define the \textbf{output degree}
    $\deg_o^T(X\to Y)$ of the output $X\to Y$
    of $T$ to be the integer triple $(a,b,c)$ such that the relation
    described in Proposition~\ref{prop:finstateslinearoutput},
    $\deg(\gamma) \leq a \deg(\alpha)+b \deg_m(\beta)+c$, holds for
    all possible $\alpha$ and $\beta_1,\dots,\beta_m$ (and the
    $\gamma$ produced by these inputs). Given the minimal linear
    term $a_{\min}$ (at least $1$) that the degree of any variable
    in any state in the input state space of $T$, and the maximal
    linear term $a_{\max}'$ and
    constant term $b_{\max}'$ for the degree of the variable in the output
    state space corresponding to $X\to Y$, we can unambiguously
    define $\deg_o^T(X\to Y) = (a_{\max}',\lfloor a_{\max}'/a_{\min}\rfloor,
    b_{\max}')$.
\end{definition}
We can now prove some elementary sufficient conditions for
$\mathrm{IPD}_{K_N}$ for a family $\{K_N\}_{N\leq 1}$ of string
machines to grow polynomially in $N$.
\begin{proposition}
    \label{prop:linearruntimeintransducers}
    Suppose a family of string machines $\{K_N\}_{N \geq 1}$, where
    $K_1$ is composed of a single filtered
    deterministic transducer with a finite
    output state space, and for $N>1$,
    $K_N$ is created by composing some of the outputs of $K_{N-1}$
    into the inputs of a new filtered deterministic transducer
    with a finite output state space,
    such that the
    remaining uncomposed input legs of each $K_N$ are the same for
    all $N$. Given the output
    $X_N\to Y_N$ of $K_{N-1}$ composed into the primary
    input of the new transducer added to obtain $K_N$, we designate
    the transducer that that output is produced from as $T_N'$, and
    write $\deg_o^{T_N'}(X_N\to Y_N) = (a_N,b_N,c_N)$.
    Now, fixing any input states and morphisms,
    $\mathrm{IPD}_{K_N} = O(N^h)$ for some $h$, regardless
    of the value of the inputs fixed, if the following
    conditions are met:
    \begin{itemize}
    \item We have that $c_N = O(N^{k})$, for some fixed $k$.
    \item We have that $a_N+b_N = O(1)$.
    \item The number of times that
    $a_N+b_N > 1$ is $O(\log N)$.
    \end{itemize}
\end{proposition}

\begin{proof}
Omitting
our input morphisms and states,
we have that $\mathrm{IPD}_{K_N} = \mathrm{IPD}_{K_{N-1}}+\gamma_N$, where
$\gamma_N$ is the output of $K_{N-1}$ composed into the new transducer's
primary input. $\mathrm{IPD}_{K_1}$
is obviously simply the degree of our initial
primary input.
$T_N'$ either adds $O(N^k)$ degrees' worth
of material to the new transducer's
primary input compared to $\mathrm{IPD}_{K_{N-1}}$, or at most
multiplies $\mathrm{IPD}_{K_{N-1}}$ by some fixed amount
and then adds this extra degree. Since subsequent multiplications
of this sum can be bounded from above by doing all the additions first
first and then multiplying the result by the cumulative
multiplication thus far, the function
bounding $\mathrm{IPD}_{K_N}$ can be expressed as a sum
$$\sum_{i=1}^N O(1)^{O(\log i)}(n+iO(i^{k})) \leq N(N^{O(1)}(n+O(N^{k+1})))
=N^{\ell+1}n+O(N^{\ell+k+2})
$$
where $\ell$ is the constant coefficient affording the $O(\log i)$ bound,
and $n$ is the total degree of the input morphisms to $K_N$.
\end{proof}

\begin{figure}[h]
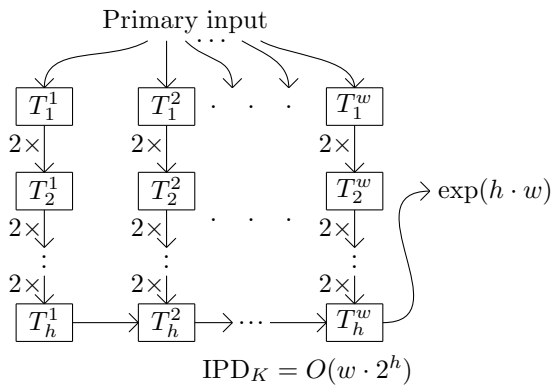

    \centering
    \ctikzfig{transducerdrawing1}
    \caption{Assuming the runtime of a state-category morphism is
    proportional to its degree, we must impose either a constant or
    logarithmic bound on the number of primary-input compositions where
    the output is multiplied to preserve runtime that is polynomial
    in the total number of transducers. In the figure, vertical
    arrows go to a transducer's primary input, while horizontal ones
    go to its auxiliary input.}
    \label{fig:boundedheightstacks}
\end{figure}

\begin{remark}
We have thus far only considered filtered transducers where the
linear term of every variable in the input state space is non-vanishing.
If we change this condition to only require that the constant term
of the variable be non-vanishing, then the total number of generators
in the $\mathbb{N}^2$-filtered category over the output category
produced in the output is still bounded by a linear function of the
primary input signature. So, the degree of the output function is now
linear in the \textit{product} of the degrees of the primary and
auxiliary inputs. One can observe, however, that any composition of
transducers will only cause outputs to be multiplied with each other
a fixed number of times, so $\mathrm{IPD}$ will now be bounded by
a polynomial in the total inputs of the string machine.
\end{remark}

\subsubsection{String machine description length}
As string machines can themselves be thought of
as morphisms in a category, and we are interested in extending
our complexity constraints to the case where the meta-vertex is
involved, we must imbue the category that string machines are
morphisms in with a filtered category structure itself. While this
could be accomplished by setting the degree of a string machine to just
be the number of transducers involved (where
a no-transducer string machine is just the identity and has degree
zero),
this obviously fails to capture the internal complexity of a transducer.
\\\\
If we require that the input category of every non-meta-vertex
transducer be finitely-generated, then we could make reference to the
size of some fixed graph representation of a transducer's transition
table and output functions corresponding to generators,
and then have this size be additive when transducers are composed.
However, as discussed in the previous section, this description
size does not constrain runtime for transducers which build other
transducers and then run them. We discuss possible approaches in
Section~\ref{sec:desccomplex}.

\subsection{The expressivity of finite-state string machines}
Now that we have obtained guarantees for the scaling of
$\mathrm{IPD}$, which can be thought of as an analogue to the
runtime of a string machine, in both the number of transducers
in a string machine and in input size, we focus on characterizing the
expressivity of string machines consisting of
\textbf{finite-state deterministic
transducers}, that is, filtered deterministic transducers where
the state space corresponding to every element of the input category
(not only the output state space)
is finite.
\\ \\
In the following section, we will assume that all transducers
have an input category that is of a form described as follows:
\begin{definition}
    \label{def:tapecategory} A \textbf{tape category} is a
    filtered-morphism
    copy-discard category $\mathcal{X}$
    freely generated over a single
    object $X$ and a finite number of morphisms of the form
    $X^m \to X^n$ (where superscripts denote monoidal self-products),
    where every generator morphism has positive degree. 
    If all our generator morphisms are of signature
    $X^m\to X$, then
    a tape
    category is the copy-discard category freely generated over the
    symmetric monoidal category freely generated from the multicategory
    with one object
    freely generated from those generators,
    by way of the functor from multicategories
    to monoidal categories proven by Hermida~\cite{hermida2000representable}
    to be an adjoint to the forgetful functor in the other
    direction.
\end{definition}
We can treat tapes as a special case of tape category, where
all generators are endomorphisms of the generating object.
\subsubsection{Without a meta-vertex}
Finite-state string machines, even without a meta-vertex, can prepend
to their output, which is not something the traditional definition
of finite-state transducer can do. It is natural to
investigate if this makes the
framework more expressive than a
deterministic finite automaton, e.g., having it
recognize palindromes by reversing a string and then feeding the product
of these two (their overlaying) into a transducer that then checks if
each character matches up.
\\ \\
Unfortunately, given morphisms
$f:A\to B$ and $g:C\to D$ in an input category (not necessarily
a tape category)
$\mathcal{C}$, and their product $f\otimes g$, we can see that,
since $F(f\otimes g) = F(f\otimes \mathrm{id}_D) \circ F(\mathrm{id}_A
\otimes g)=F(\mathrm{id}_B \otimes g)\circ F(f \otimes \mathrm{id}_C)$,
a state transition on a collection of finite state spaces induced 
by the product of two morphisms can only capture the same sort of information
that could be captured by independently running two finite-state
transducers on both morphisms and comparing their states at the very end.
\\ \\
Thus, ``per-character'' comparisons of the sort that would require
conditions analogous to
aligning characters on tape would require infinite state spaces to
represent the equivalent memory operations being performed. However, due
to the nature of the auxiliary input, as
well as the aforementioned prepending
capabilities, we are clearly able to model
functions that finite state transducers cannot, even without a meta-vertex.
\\ \\
To characterize the class of decision
problems computable by non-meta-vertex
finite-state string machines, we first prove the following lemma:
\begin{lemma}
    \label{lem:fsmprepending} Suppose some finite-state transducer
    $T$ with
    primary input signature
    $X\to X$ in a tape category $\mathcal{X}$, structure
    functor $F$, and output state space
    $S(A)$, such that these are its only inputs and outputs.
    Fix a finite set of morphisms $\{g_1,\dots,g_n:X\to X\}$ in
    the freely-generated $\mathbb{N}^2$-filtered category over
    $\mathcal{X}$ with a single generating signature
    $X\to X$ of degree $x$, such that the linear coefficient of
    the degree of each $g_i$ is at most $1$ (that is, the generator
    corresponding to $X\to X$ appears at most once).
    Suppose
    we are able to construct a
    primary input morphism $g:X\to X$
    for this transducer in the following
    way:
    \begin{itemize}
        \item Our starting value for $g$ is $\mathrm{id}_X$.
        \item We can replace $g$ with
        with one of the previously-given $g_i$,
        with $g$ substituted for the
        generator morphism. We can perform this step a finite
        of times.
    \end{itemize}
    Then, given some state transition reached as a result of some
    such $g$,
    the amount of memory required to compute the output state of
    the transducer is $O(1)$ in the degree of $g$.
\end{lemma}
\begin{proof}
We will assume that one ``computational step'' involves performing one
of the operations involved in constructing $g$ described above. We
will demonstrate that computing the next state after
any of the operations performed requires storing an amount of information
that is $O(1)$ in the degree of $g$.
\begin{itemize}
\item If one of the $g_i$ does not have the generator morphism appear
at all, then $S(A)$ after this $g_i$ has replaced $g$ depends solely
on the $g_i$ used, and since the set of such $g_i$ is finite, it takes
a finite amount of space to store the state reached from the starting
state when one of these $g_i$ are evaluated. 
\item If one of the $g_i$ only involves post-composing an endomorphism
of $X$ into the generator, then storing the state reached from a given
state after any of these endomorphisms is post-composed takes a fixed
amount of memory, since the set of $g_i$ is finite.
\item If one of the $g_i$ involves only pre-composing and post-composing
endomorphisms of $T$ into $g$, then computing the resultant state after
the action of one of these $g_i$ only requires us to store $|S(A)|$
extra ``possible current states.'' That is, if we know what state we
would be in after the evaluation of $g$ had we started at any other
one of the $|S(A)|$ states of the transducer, then after computing
which new starting state we would have started at (which, since the
set of endomorphisms we pre-compose is finite, is equivalent to storing
a fixed-size transition table), we can then pick
the state that the copy of the transducer that started at that state
would be at after the action of $g$. This is therefore equivalent
to running $|S(A)|$ copies of a finite state automaton, where each
copy starts at a different state .
\item Suppose one of the $g_i$ is obtained by taking a morphism
$\alpha:X\to X^n$, post-composing morphisms into each branch (such that
one branch includes the generator), and then finally post-composing
a morphism $\beta:X^n \to X$. Even if
multiple ``branchings'' occur, since the generator where $g$ will
be substituted appears only once, we can still decompose this
into a morphism obtained by
first taking the monoidal product of $g$ with a number of fixed
morphisms such that the overall signature of their product is
$X^n \to X^n$, and then composing in $\alpha$ and $\beta$.
Now, we only need
$O(1)$ space to store the functions $S(F(X)) \to S(F(X^n))$ and
$S(F(X^n)) \to S(F(X))$ induced by $\alpha$ and $\beta$, since these
are functions between finite sets. Additionally, for the function
$S(F(X^n)) \to S(F(X^n))$ induced by the monoid product of $g$
with with a fixed morphism
$\gamma$ of signature $X^{n-1} \to X^{n-1}$, by
the property discussed at the start of this section, we only need to
be keeping track of the action of $F(g \otimes \mathrm{id}_{X^{n-1}})$
on $S(F(X^n))$ (which, since there are a finite number of
$g_i$, only requires keeping an absolutely bounded number of automata
running), and then compute the action of
$F(\mathrm{id}_X\otimes \gamma)$ on the result, which, as a function
with a finite domain and codomain, takes finite space to store.
\end{itemize}
Therefore, to compute the state transition induced by $g$, we only
need to keep $|S(A)|$ copies of the finite state automaton running,
as well as some fixed number of
finite-state automata with input and output
state spaces corresponding to
$S(F(X^n))$ for some fixed $n$.
\end{proof}
With this lemma, it now remains to show, roughly speaking, that
for any string machine made out of
finite-state
deterministic transducers which decides some language (represented by
morphisms in a tape category), the computation performed at every
input morphism can be broken into a constant-bounded number of transitions
of the type described in the lemma. We present a conjecture that
finite-state string machines without a meta-vertex can only recognize
regular languages, followed by a sketch for a proof of this conjecture.
\begin{conjecture}
    \label{prop:constantspacefortapecategories}
    Suppose a string machine composed only of finite-state transducers,
    where the input and output categories of each transducer are
    tape categories. Suppose this string
    machine has only one free input $X\to X$ in
    some tape category $\mathcal{X}$ where all generating
    morphisms (which form a set
    $\Sigma$) are endomorphisms of
    the generating object $X$, and has as its
    sole output a state space $S(A)$, with some subset of $S(A)$
    designated as accepting. Then, taking an endomorphism
    $X\to X$ in
    $\mathcal{X}$ which does not include any instances
    of the copy morphism to be a string over the alphabet $\Sigma$,
    any language over $\Sigma$
    accepted by this string machine is regular.
\end{conjecture}
\begin{proofsketch}
It suffices to prove that the computation involved in deciding
this language takes constant space as a function of input size,
as $\mathrm{DSPACE}(O(1))=\mathrm{REG}$. More precisely, given
any input morphism $g:X\to X$ and a finite set of
possible character strings that
we can post-compose into $g$, we
desire to show that there is a constant-bounded amount of information
(independent of the size of $g$) that
we need to store, which will allow us to calculate the output state after
any of these character strings are added to $g$.
By Proposition~\ref{prop:boundedduplication}, each transducer
involved in the string machine can create a constant-bounded
number of duplicates of its auxiliary inputs within each of
its outputs. This means that to predict the state of a given
transducer after one of the operations described in
Lemma~\ref{lem:fsmprepending} done to an auxiliary input,
we need to store a constant amount of information. Moreover, due to the limitation
caused by the size of a variable on the number of transitions one
can make before we run out of variable space and must either limit
ourselves to Lemma~\ref{lem:fsmprepending}-like operations or
reset, the number of ``shapes'' auxiliary inputs can exist in
is absolutely bounded. The hope is that if we back-propagate until
reaching the primary input of the first transducer, we will still
have a constant-bounded amount of information to keep track of to
predict the output state after the next character is reached.
\end{proofsketch}
\begin{remark}
    This conjecture is already
    suggested by the combination of Corollary~\ref{corr:linearruntime}
    and classical results by Kobayashi~\cite{kobayashi1985onetape}
    and Hennie~\cite{hennie1965onetape} which together imply that
    any one-tape Turing machine that runs in $o(n\log n)$ time decides
    a regular language. The reason
    the proposition is not trivially implied by
    these results is that a naive simulation of the string machine
    in question on a single-tape Turing machine will cross the last cell
    of its initial input a guaranteed $O(n)$ times (when going back and
    forth to read the next cell to write the inputs to feed to the next
    automaton in line). The conjecture would also automatically
    be implied if we prove that the composition of two finite-state
    transducers can be equivalently modeled by a single finite-state
    transducer.
\end{remark}

\subsubsection{With a meta-vertex}
The meta-vertex, even at meta-level one, can clearly be seen to
increase the expressivity of finite-state string machines. We
construct a string machine which can be recognize
palindromes as follows:
\\ \\
Given
a tape category $\mathcal{X}$ where all generating morphisms are
endomorphisms of the generating object $X$, mark 
a special generating endomorphism $r$. Given a generating
endomorphism $\alpha\neq r$
of $X$, we define $T_{\alpha}$ to be the
filtered deterministic transducer with input and output category
$\mathcal{X}$, with input and output signatures $X\to X$ and three
states, a fixed one of which is the starting state.
When in the starting state, $T_{\alpha}$ will transition to the
``accepting state'' and store $\mathrm{id}_X$ in the sole variable
at that state if it encounters $\alpha$, and transition to the
``rejecting state'' and store $r$ if any other
generating morphism is encountered. When in the rejecting state, all
morphisms are mapped to the identity on the rejecting state. When
in the accepting state, all morphisms transition back to the accepting
state, but this time append the current primary input generator
morphism to the output. The end result is that when
given an endomorphism $s:X\to X$, the transducer $T_{\alpha}$ will
output $s$ with its first ``character'' removed if and
only if it starts with $\alpha$, and output $r$ otherwise.
As a morphism, 
it is of signature $(X\to X)_{\mathcal{X}} \to (X\to X)_{\mathcal{X}}$.
\\ \\
Now, we define a
filtered
deterministic transducer $T$ with input category $\mathcal{X}$, which
outputs a filtered deterministic transducer without meta vertices.
This transducer has input signature $X\to X$ and only one state,
with one variable of signature $(X\to X) \to (X \to X)$, and
it is assumed that no morphism $X\to X$ fed to $T$ will include
$r$. Given a generator morphism $\alpha$ of $\mathcal{X}$,
the transducer
$T$ will prepend $T_{\alpha}$ to the current variable's stored string
machine. If we compose a meta vertex that takes the output of $T$
and takes the same input morphism as $T$ and passes it on to the
string machine formed by $T$, then we obtain a string machine that
outputs $\mathrm{id}_X$ if it is given a palindrome, and
$r$ if it is fed anything else.
\section{Conclusions}
Thus far, we have proven results about the runtime and expressivity of
finite-state deterministic transducers in the string machines framework.
This framework allows us to model the compositionality of transducers
in a way that allows for more compression of description length.
For example, we
are able to implement the
trivial time-space tradeoff of recognizing a regular
language that is the intersection of several regular languages by
running several automata in sequence instead of the product automaton,
or the cascade product in Krohn-Rhodes
theory~\cite{krohn1965algebraic} by having a transducer ``shuffle''
its primary input using its ability to prepend.
\\ \\
Additionally, our runtime scaling
guarantees in the number of transducers in 
a string machine translate to sufficient conditions for
transducers which build and run other transducers in runtime to run in
polynomial time.
This will allow
us to take advantage of the additional compression opportunities the
meta-vertex presents.
These results are valuable for developing
string machines into a framework that can be used to model hypothesis
spaces representable by automata.
\subsection{Further work}
\subsubsection{Description complexity that encapsulates runtime}
\label{sec:desccomplex}
As exemplified by Definition~\ref{def:freegenfilteredmorphism},
the construction of filtered sets can be extended to a
partially-ordered monoid that is not $\mathbb{N}$. Since a string
machine without a meta-vertex will have the degrees of its outputs
be bounded by a linear function of the degrees of its inputs, we
can assign a degree $ax+b$ to finite-state
filtered transducers, where $a$ and $b$
are the maximal values for some sum of the coefficients appearing
in the output degree in Definition~\ref{def:outputdegree}, where
the partial order is the same as in 
Definition~\ref{def:freegenfilteredmorphism}. However, in this case,
the monoid operation imposed on the polynomials is composition, due
to how the linear relations bounding output size get composed when
transducers are composed, and the filtered-set category is no
longer symmetric monoidal since this operation is not commutative.
\\ \\
Now, the variables in our
state categories would have their degrees be of the form
$ay+bx+c$, where a morphism of degree $dx+e$ can compose
in a transducer of at most that degree, meaning the output
variable is of degree $a(y+(dx+e))+bx+c=ay+(dx+e)^a \circ
bx+c$ by abuse of notation. where the superscript indicates
$a$-fold self-composition. At further meta-levels, we keep
adding terms, though it is unclear
how we would disambiguate which variable to compose
into for the monoid operation. 
\subsubsection{Characterizing the
expressivity of transducers with a meta-vertex
and monadic transducers}
We have not yet established bounds on the expressivity of
finite-state deterministic transducers which incorporate the
meta-vertex, but, due to the bounding on levels of meta, it is
likely that the languages they decide are some subset of context-free
languages. The runtime guarantees also need to be translated to
monadic transducers.
Work on bounding the expressivity of monadic string
machines may involve defining equivalents
for string machines of various
properties~\cite{mohri2004weighted, allauzen2004optimal, 
kostolanyi2022determinisability, bell2023computing} which
characterize determinizable weighted transducers, which may in turn
allow us to bound the computational complexity of computing the
distribution on output morphisms given by a monadic transducer.
\bibliographystyle{abbrv}
\bibliography{main}

\end{document}